\documentclass[a4paper,twocolumn,11pt,unpublished]{quantumarticle}
\pdfoutput=1
\usepackage[utf8]{inputenc}
\usepackage[T1]{fontenc}
\usepackage[english]{babel}
\usepackage{amsmath,amssymb,amsthm,bm}
\usepackage{graphicx,booktabs,enumitem,microtype,xcolor}
\usepackage{hyperref}
\hypersetup{colorlinks=true,allcolors=quantumviolet,
 pdftitle={Record consistency under repeated readout and apparatus composition},
 pdfauthor={D.M. Theshan N. Weerasinghe}}
\makeatletter
\let\quantum@original@printtitle\@printtitle
\renewcommand*\@printtitle{{\color{quantumviolet}\quantum@original@printtitle}}
\makeatother
\newtheorem{theorem}{Theorem}[section]
\newtheorem{lemma}[theorem]{Lemma}
\newtheorem{proposition}[theorem]{Proposition}
\newtheorem{corollary}[theorem]{Corollary}
\theoremstyle{definition}

\newtheorem{postulate}[theorem]{Assumption}
\theoremstyle{remark}

\newcommand{\TV}{\operatorname{TV}}
\newcommand{\sech}{\operatorname{sech}}
\newcommand{\Tr}{\operatorname{Tr}}
\newcommand{\E}{\mathbb{E}}

\newcommand{\R}{\mathbb{R}}
\newcommand{\dd}{\,\mathrm{d}}
\AtBeginDocument{\normalfont}
\begin{document}
\title{Record consistency under repeated readout and apparatus composition}
\author{D.M. Theshan N. Weerasinghe}
\email[\newline]{theshan.dissanayakemudiyanselage1@monash.edu}
\affiliation{Department of Electrical and Computer Systems Engineering, Monash University, Clayton, Australia}
\date{}
\maketitle
\begin{abstract}
No-signalling constraints on isolated outcome probabilities do not determine whether an apparatus preserves the joint law of its accessible records under repeated readout, network composition, or reuse. We study these extensions through a conditional-record criterion and four explicit models. For a wrapped-Cauchy retrocausal source, protection under two conditionally independent continuous readouts forces an angle-independent response; three readouts suffice without continuity. A finite angular-bandwidth assumption gives a quantitative leakage bound, whereas continuity alone gives no uniform stability. At a specified two-source Bell junction, the required correction matrix has full nonnegative rank on every even uniform angular grid. A Gaussian controller can cancel branch-mass dependence while leaving a distinguishable joint record, for which we derive exact noise and law-change tradeoffs. Finally, a fixed singlet with shared classical noise reproduces the variable-width isolated outcome family, but retained and refreshed controllers generate different temporal statistics. We prove a finite-resource obstruction and construct positive finite-controller approximations. Numerical studies comprising 8.3 million independent simulated histories support the model-specific predictions. Together, the results identify conditions that must accompany marginal no-signalling when extending an apparatus to richer records and interventions.
\end{abstract}

\section{Introduction}\label{sec:introduction}
An outcome distribution can satisfy no-signalling even when its underlying history law depends on a later measurement setting. Retrocausal models exploit this distinction by imposing boundary conditions on histories rather than requiring all hidden variables to be generated in temporal order. In the spin models studied by Almada \emph{et al.}~\cite{almada}, symmetry removes observable setting dependence from isolated outcome marginals. The relation between such symmetries and the causal explanation of no-signalling is central to assessing the explanatory content of a model~\cite{wood}. An additional question arises once a detector can be read repeatedly, a second source is connected, or a compensating device leaves its own observable record: which features of the isolated construction survive these extensions?

The relevant object is then a joint record distribution. Correcting the probability of a detector outcome need not correct the conditional distribution of its associated pointer, memory, or auxiliary controller. Likewise, specifying a normalized distribution for every fixed setting does not by itself specify how raw history weights are to be connected through adaptive control. Network scenarios make the distinction between independent preparation and conditional dependence explicit~\cite{branciard,price}, while quantum-network and process descriptions provide established frameworks for treating multitime operations and memory~\cite{chiribella,pollock,taranto}.

This article develops an apparatus-level analysis of these issues, starting from the wrapped-Cauchy source and detector-record conditions of Refs.~\cite{weerasinghe1,weerasinghe2}. All definitions needed below are given explicitly. The organizing criterion separates correction of branch masses from invariance of the conditional law of additional records. Four models then isolate distinct extension requirements: passive access to a fixed hidden angle, factorization at a two-source junction, accessible degrees of freedom of a compensating controller, and controller reset in a variable-width operational construction.

\subsection{Results and assumptions}
For repeated passive readout, the hidden angle is held fixed and detector draws are independent conditional on that angle. Under these premises, two continuous readouts or three merely measurable readouts can be protected for every future axis only if the detector response is angle independent. A bound on the angular bandwidth of its antipodal contrast yields a nonzero lower bound on two-record leakage. A smooth counterfamily establishes why continuity alone does not yield such a uniform quantitative conclusion.

For network composition, we specify a positive Bell-junction kernel and compare its branch masses with a quantum benchmark. Two independent scalar corrections cannot realize the benchmark, and the correction matrix has rank and nonnegative rank equal to the size of every even uniform angular grid. The result concerns this kernel and this product-sum correction representation; it does not impose the same resource bound on arbitrary network models.

For readable compensation, we consider a two-dimensional Gaussian boundary ensemble. An inverse partition factor can protect both the detector-bit marginal and the controller-coordinate marginal while their joint law remains setting dependent. We derive its exact total variation, a finite-noise tradeoff, a sharp unrestricted conditional-law repair bound, and the compensation parameters that protect the full record.

For variable widths, a fixed singlet and a shared hyperbolic-secant classical controller reproduce the isolated outcome family. Fresh local controllers yield a specified measure-and-prepare continuation instrument, whereas retention changes temporal statistics. A Cauchy-matrix argument excludes a fixed finite \emph{total} shared resource for the exact continuum family under width-local access; a positive discrete construction gives uniform finite-error approximations.

These models have distinct history laws and preparation assumptions. The first retains the source's conditional hidden-angle law, the junction introduces an additional raw interaction, the Gaussian ensemble postulates a source--controller coupling, and the variable-width model is an operational quantum construction. Their results therefore give conditional requirements for apparatus extension, rather than a single microscopic realization. The mathematical statements are supported by deterministic cross-checks and four computational experiments; the sampling methods and uncertainty conventions are collected in Appendix~\ref{app:methods}.

\section{Record consistency}\label{sec:framework}
\subsection{Source and apparatus specification}
For finite positive width $\gamma$, define the wrapped-Cauchy kernel
\begin{equation}\label{eq:wrapped}
 p_\gamma(t)=\frac{\sinh\gamma}{2\pi(\cosh\gamma-\cos t)},
 \quad t\in\R/(2\pi\mathbb Z).
\end{equation}
For outcomes $a,b\in\{-1,+1\}$, let $\theta_{x,a}=x+\pi\mathbf1_{a=-1}$ and set
\begin{equation}\label{eq:sourcefactors}
\begin{aligned}
 A_a(\lambda)&=p_{\gamma_A}(\theta_{x,a}-\lambda),\\
 B_{y,b}(\lambda)&=p_{\gamma_B}(\theta_{y,b}-\lambda-\pi),\\
 R_A(\lambda)&=\sum_a A_a(\lambda),\qquad
 R_B(\lambda;y)=\sum_b B_{y,b}(\lambda).
\end{aligned}
\end{equation}
The normalized isolated history law is
\begin{equation}\label{eq:sourcehistory}
\begin{aligned}
 P_0(\dd\lambda,a,b\mid x,y)
 &=\frac{A_a(\lambda)B_{y,b}(\lambda)}{Z_{xy}}\dd\lambda,\\
 Z_{xy}&=\int_0^{2\pi} R_A(\lambda)R_B(\lambda;y)\dd\lambda.
\end{aligned}
\end{equation}
Integrating over $\lambda$ yields
\begin{equation}\label{eq:isolatedlaw}
 P_0(a,b\mid x,y)=\frac{1-ab\sech(\gamma_A+\gamma_B)\cos(x-y)}4.
\end{equation}
Matching Eq.~\eqref{eq:isolatedlaw} fixes outcome probabilities but does not identify the underlying law in Eq.~\eqref{eq:sourcehistory}.

The apparatus specification consists of a raw history measure, a boundary preparation, permitted readout kernels, retained memory, and intervention rules. The following assumptions fix how these objects are compared.
\begin{postulate}[Accessible records]
The earlier record is the joint tuple of outputs permitted by the apparatus before the intervention being compared. A future-dependent hidden variable is included only when an admissible readout makes it accessible.
\end{postulate}
\begin{postulate}[Fixed preparation and interventions]
The preparation and earlier apparatus are held fixed across compared interventions. An adaptive intervention is a specified measurable function of an available record. All outcomes, including any failure outcomes, are retained.
\end{postulate}
\begin{postulate}[Passive access]
A passive readout leaves the stipulated hidden history law unchanged. Repeated independent passive readout uses fresh detector randomness conditional on the same unresolved hidden state. Copying a recorded value and preparing a disturbed state are different operations.
\end{postulate}
\begin{postulate}[Resource accounting]
Preparation independence is distinguished from dependence after conditioning. Any finite total resource bound includes shared classical registers and environmental memory. Discarding a register from the description does not reset it.
\end{postulate}

We call a record \emph{protected} when its distribution is identical for all compared future interventions. We use
\begin{equation}
 \TV(P,Q)=\sup_E|P(E)-Q(E)|
\end{equation}
for total variation; for densities it equals one half of their integrated absolute difference.

\subsection{Conditional laws and adaptive normalization}
Let $h$ be a finite earlier branch label and $u$ a later intervention. Disintegrate the unnormalized history measure as
\begin{equation}
 R_u(h,\dd z)=r_u(h)Q_u(\dd z\mid h),
 \quad 0<r_u(h)<\infty,
\end{equation}
where $Q_u$ is normalized. A scalar correction multiplies each branch by $c_u(h)>0$. For a permitted recorded output with kernel $L_u(\dd e\mid h,z)$, define
\begin{equation}
 \widehat Q_u(\dd e\mid h)
 =\int Q_u(\dd z\mid h)L_u(\dd e\mid h,z).
\end{equation}
Dependence on $u$ denotes a context-dependent history law; it does not assign an ordinary earlier device a temporally forward input from the future.

\begin{theorem}[Conditional-law criterion]\label{thm:conditional}
Suppose the corrected branch distribution is a fixed $q(h)$ for every $u$. The joint record $(h,e)$ is protected if and only if $\widehat Q_u(\cdot\mid h)$ is independent of $u$ for every $h$ with $q(h)>0$. Moreover,
\begin{equation}\label{eq:jointtv}
\begin{aligned}
 &\TV(P_u^{H,E},P_v^{H,E})\\
 &\quad=\sum_h q(h)\TV\!\left(\widehat Q_u(\cdot\mid h),\widehat Q_v(\cdot\mid h)\right).
\end{aligned}
\end{equation}
\end{theorem}
\begin{proof}
The corrected record law is $P_u(h,\dd e)=q(h)\widehat Q_u(\dd e\mid h)$. Distinct values of $h$ define disjoint measurable components, so total variation adds over $h$, giving Eq.~\eqref{eq:jointtv}. The nonnegative terms vanish precisely under the stated condition.
\end{proof}
This disintegration identity separates branch-mass correction from protection of additional records. The following normalization identity is useful for defining comparison laws but imposes no locality requirement.
\begin{lemma}[Branchwise normalization]\label{lem:normalization}
Let $n_h$ be target probabilities and $R_h$ raw branch measures with finite masses $d_h>0$ on their support. Then $P(h,\dd z)=n_hR_h(\dd z)/d_h$ is normalized and has marginal $n_h$.
\end{lemma}
\begin{proof}
Integrating each branch gives $n_h$; summing over $h$ gives one.
\end{proof}
Whether the factor $n_h/d_h$ can be generated by an allowed interaction is a separate question, examined in Secs.~\ref{sec:network} and~\ref{sec:controller}.

\begin{lemma}[Adaptive common mass]\label{lem:adaptive}
Let $k$ be a finite earlier record with target law $q(k)>0$. Suppose the integrated corrected raw mass at fixed $u$ is $D_u(k)=z_uq(k)$, where $z_u>0$. Substituting $u=f(k)$ in the raw law and then normalizing globally gives
\begin{equation}\label{eq:adaptivecommon}
 q_f(k)=\frac{z_{f(k)}q(k)}{\sum_jz_{f(j)}q(j)}.
\end{equation}
Equality of all $z_u$ suffices for every policy to preserve $q$. It is also necessary when at least two records have positive probability and controls can be chosen freely on those records.
\end{lemma}
\begin{proof}
Summing over unrecorded outputs after substitution gives Eq.~\eqref{eq:adaptivecommon}. Equal $z_u$ cancel. Conversely, assign two allowed controls to two positive-probability records and compare their probability ratio with $q(k)/q(j)$. Preservation requires the corresponding masses to agree.
\end{proof}
Fixed-context no-signalling and closure under raw adaptive substitution are therefore distinct conditions. A construction using normalized conditional transitions specifies a different composition rule and must be assessed with that rule held fixed.

\section{Repeated passive readout}\label{sec:repeated}

Protection of a single detector output need not extend to the joint record of repeated use. We examine this question for a binary detector coupled to one unresolved source angle, while retaining the source law and leaving that angle unchanged. The distinction between another independent readout and a copy of an existing record is essential.

Fix an earlier axis and outcome $(x,a)$, with $a\in\{+1,-1\}$, and positive finite widths $\gamma_A,\gamma_B$. On the circle $\mathbb T=\mathbb R/(2\pi\mathbb Z)$, define
\begin{equation}
 p_\gamma(t)=\frac{\sinh\gamma}
 {2\pi(\cosh\gamma-\cos t)}.
\end{equation}
Writing $\theta_{x,a}=x+\pi\mathbf 1_{a=-1}$, the source factors are
\begin{equation}
\begin{aligned}
 A(\lambda)&=p_{\gamma_A}(\theta_{x,a}-\lambda),\\
 R_B(\lambda;y)&=p_{\gamma_B}(y-\lambda)\\
              &\quad+p_{\gamma_B}(y-\lambda-\pi).
\end{aligned}
\end{equation}
The conditional source measure is
\begin{equation}\label{eq:repeated-source}
\begin{aligned}
 \mu_y(\dd\lambda)
   &=\frac{A(\lambda)R_B(\lambda;y)}{D(y)}\dd\lambda,\\
 D(y)&=\int_0^{2\pi} A(\lambda)R_B(\lambda;y)\dd\lambda.
\end{aligned}
\end{equation}
Here the future axis ranges over all $y\in[0,\pi)$. Appending the detector preserves Eq.~\eqref{eq:repeated-source}. Its response is $P(R=+\mid\lambda)=f(\lambda)$, with $0\le f\le1$ and no explicit dependence on $y$. Repeated readouts have independent detector randomness conditional on the same unchanged $\lambda$. Statistical conditioning on an observed bit is allowed; a physical update or reset of the hidden angle is not part of this model.

\subsection{Exact protection and finite error}

Set $A_\pi(\lambda)=A(\lambda+\pi)$ and $p=A/(A+A_\pi)$. The relevant symmetry acts on the weighted antipodal pair, rather than on either response separately.

\begin{lemma}[Antipodal characterization]\label{lem:antipodal}
The mean of an integrable response $g(\lambda)$ is the constant $m$ for every future axis if and only if, almost everywhere,
\begin{equation}\label{eq:repeated-antipodal}
 p(\lambda)g(\lambda)
 +[1-p(\lambda)]g(\lambda+\pi)=m.
\end{equation}
\end{lemma}
\begin{proof}
The condition is
\begin{equation}
 \int_0^{2\pi} A(\lambda)[g(\lambda)-m]
 R_B(\lambda;y)\dd\lambda=0
\end{equation}
for every $y$. The Fourier expansion
\begin{equation}
 p_\gamma(t)=\frac1{2\pi}
 \sum_{k\in\mathbb Z}e^{-|k|\gamma}e^{ikt}
\end{equation}
shows that the translates of $R_B$ have nonzero coefficients at every even harmonic. Every even coefficient of $A(g-m)$ therefore vanishes. Uniqueness of Fourier coefficients gives its zero $\pi$-periodization. Division by $A+A_\pi>0$ yields Eq.~\eqref{eq:repeated-antipodal}; antipodal pairing proves the converse. This is the source-specific characterization of Ref.~\cite{weerasinghe2}, included here for the repeated-readout analysis.
\end{proof}

For conditionally independent readouts, the probability of $j$ consecutive $+$ outcomes is $\int f^j\dd\mu_y$. Protection of the full $n$-readout law thus implies Eq.~\eqref{eq:repeated-antipodal} for $g=f^j$, $1\le j\le n$. Conversely, these moment conditions protect every $n$-bit pattern, since $f^k(1-f)^{n-k}$ is a polynomial of degree at most $n$.

\begin{theorem}[Two continuous readouts]\label{thm:two}
If $f$ is continuous and the complete two-readout distribution is independent of every future axis, then $f$ is constant. Conversely, an angle-independent response is protected under every independent repetition.
\end{theorem}
\begin{proof}
Joint protection includes marginal protection. Lemma~\ref{lem:antipodal} applied to $f$ and $f^2$ gives constants $m,s$ such that
\begin{equation}
\begin{aligned}
 pf+(1-p)f_\pi&=m,\\
 pf^2+(1-p)f_\pi^2&=s,
\end{aligned}
\end{equation}
where $f_\pi(\lambda)=f(\lambda+\pi)$. Consequently,
\begin{equation}\label{eq:repeated-variance}
 s-m^2=p(1-p)(f-f_\pi)^2.
\end{equation}
Continuity extends these identities from almost everywhere to every angle. The continuous antipodally odd function $f-f_\pi$ vanishes somewhere by the intermediate value theorem. The constant on the left of Eq.~\eqref{eq:repeated-variance} is therefore zero. Since $p(1-p)>0$, one obtains $f=f_\pi=m$ everywhere. A constant response gives an axis-independent product law for any number of repetitions.
\end{proof}

\begin{theorem}[Three measurable readouts]\label{thm:three}
If the complete law of three conditionally independent identical readouts is protected, every measurable response $0\le f\le1$ equals a constant almost everywhere.
\end{theorem}
\begin{proof}
The first three antipodal mixture moments are constant. Let $m$ denote the mean, $V$ the second central moment, $T$ the third central moment, and $d=f-f_\pi$. Then
\begin{equation}
\begin{aligned}
 V&=p(1-p)d^2,\\
 T&=p(1-p)(1-2p)d^3.
\end{aligned}
\end{equation}
If $V>0$, these identities imply
\begin{equation}
 \frac{T^2}{V^3}=\frac{(1-2p)^2}{p(1-p)}
 \quad\text{almost everywhere}.
\end{equation}
For the specified source,
\begin{equation}
 p(\lambda)=\frac12\left[
 1+\sech\gamma_A\cos(\theta_{x,a}-\lambda)\right].
\end{equation}
At positive finite width, $p$ varies over a nondegenerate interval. The preceding ratio is a nonconstant continuous function of angle and cannot equal a constant almost everywhere. Thus $V=0$, and Eq.~\eqref{eq:repeated-antipodal} gives $f=m$ almost everywhere.
\end{proof}

Three grouped observations also identify a two-component mixture in established statistical results~\cite{vandermeulen}. The additional source-specific step here is that the varying antipodal mixture weight excludes a nontrivial protected response.

To quantify departures from protection, write $P_y^{(2)}$ for the full two-bit distribution and define
\begin{equation}\label{eq:repeated-leakage}
\begin{aligned}
 \mathcal L_2&=\sup_{y,y'}\TV(P_y^{(2)},P_{y'}^{(2)}),\\
 \langle d^2\rangle
  &=\frac1{2\pi}\int_0^{2\pi}d(\lambda)^2\dd\lambda.
\end{aligned}
\end{equation}
Here $\TV(P,Q)=\tfrac12\sum_r|P(r)-Q(r)|$ and $\Gamma=\gamma_A+\gamma_B$.

\begin{theorem}[Finite angular bandwidth]\label{thm:bandwidth}
Assume exact one-record protection. If the antipodal contrast $d=f-f_\pi$ is a trigonometric polynomial of angular degree at most $J\ge1$, then, with $q_J=e^{-2(J+1)\Gamma}$,
\begin{equation}\label{eq:bandwidthbound}
 \mathcal L_2\ge
 \tanh\gamma_A\,\langle d^2\rangle
 \frac{q_J}{\sqrt{1-q_J^2}}.
\end{equation}
The degree constraint applies to $d$, not necessarily to $f$ itself.
\end{theorem}
\begin{proof}
Rotate coordinates so that $\theta_{x,a}=0$ and set $H=A+A_\pi$. Direct algebra gives
\begin{equation}
 Hp(1-p)=\frac{\tanh\gamma_A}{4\pi}.
\end{equation}
One-record protection implies $f=m+(1-p)d$. Hence the excess double-plus probability satisfies
\begin{equation}
\begin{aligned}
 v(y)&:=P_y(++)-m^2=\frac{N(y)}{D(y)},\\
 N(y)&=\frac{\tanh\gamma_A}{4\pi}
 \int_0^\pi d(\lambda)^2R_B(\lambda;y)\dd\lambda.
\end{aligned}
\end{equation}
The denominator can equivalently be written as $D(y)=\int_0^\pi H R_B\dd\lambda$. In the Fourier basis $e^{2iny}$, convolution yields
\begin{equation}
 D_0=\frac1\pi,\qquad
 |D_n|=\frac{e^{-2|n|\Gamma}}\pi.
\end{equation}
Since $d^2$ has angular degree at most $2J$, $N_n=0$ for $|n|>J$. Its mean gives
\begin{equation}
 c:=\frac{N_0}{D_0}
   =\frac{\tanh\gamma_A}{4}\langle d^2\rangle.
\end{equation}
Equip $y\in[0,\pi)$ with the probability measure
\begin{equation}
 \nu(\dd y)=\frac{D(y)}{\pi D_0}\dd y.
\end{equation}
For $n=J+1$ and $U=e^{-2iny}$, the first absent numerator harmonic gives
\begin{equation}
 \mathbb E_\nu v=c,\qquad
 |\mathbb E_\nu U|=q_J,\qquad
 \mathbb E_\nu(vU)=0.
\end{equation}
Cauchy--Schwarz therefore implies
\begin{equation}
\begin{aligned}
 c q_J
 &=\left|\mathbb E_\nu
       [(v-c)(U-\mathbb E_\nu U)]\right|\\
 &\le\sqrt{\operatorname{Var}_\nu v}\sqrt{1-q_J^2}\\
 &\le\frac{\operatorname{osc}v}{2}\sqrt{1-q_J^2}.
\end{aligned}
\end{equation}
The last step uses the variance bound for a real variable in an interval. With the one-bit mean fixed, the four pair probabilities change by $(\Delta v,-\Delta v,-\Delta v,\Delta v)$. Thus $\mathcal L_2=2\operatorname{osc}v$, which proves Eq.~\eqref{eq:bandwidthbound}.
\end{proof}

The bound concerns the continuum of future axes. It is not asserted to be sharp. Continuity without a uniform regularity constraint provides no corresponding uniform stability statement.

\begin{proposition}[Continuity does not imply uniform stability]
At fixed positive finite widths, there are continuous, exactly one-record-protected responses with nonvanishing angle dependence and $\mathcal L_2\to0$.
\end{proposition}
\begin{proof}
Again choose coordinates with $\theta_{x,a}=0$. For
$0<c_0<\tfrac12\tanh(\gamma_A/2)$, define
\begin{equation}\label{eq:repeated-smooth}
 f_K(\lambda)=\frac12+
 c_0\tanh(K\sin\lambda)
 \sqrt{\frac{1-p(\lambda)}{p(\lambda)}}.
\end{equation}
These smooth responses lie in $[0,1]$. Their weighted antipodal mean is $1/2$ and their antipodal mixture variance is $c_0^2\tanh^2(K\sin\lambda)$, which converges to $c_0^2$ almost everywhere. On the quotient circle, the excess double-plus probability is the average of this variance against the density
\begin{equation}
 q_y(\lambda)=\frac{H(\lambda)R_B(\lambda;y)}{D(y)}.
\end{equation}
These densities on $0\le\lambda<\pi$ are uniformly bounded over the compact angle and setting domain. Consequently,
\begin{equation}
\begin{aligned}
 \sup_y|v_K(y)-c_0^2|
 &\le c_0^2\|q\|_\infty\\
 &\quad\times\int_0^\pi
   [1-\tanh^2(K\sin\lambda)]\dd\lambda.
\end{aligned}
\end{equation}
The right-hand side tends to zero by dominated convergence. It follows that $\mathcal L_2=2\operatorname{osc}v_K\to0$. In contrast, evaluation at $\lambda=\pi/2,3\pi/2$ gives $\operatorname{osc}f_K\ge2c_0\tanh K$. Thus the response variation does not vanish. The sequence has no fixed angular bandwidth.
\end{proof}

The finite-error behavior is illustrated in Fig.~\ref{fig:repeated-finite-error}; its numerical evaluation is described below.

\begin{figure}[tbp]
\centering
\includegraphics[width=\columnwidth]{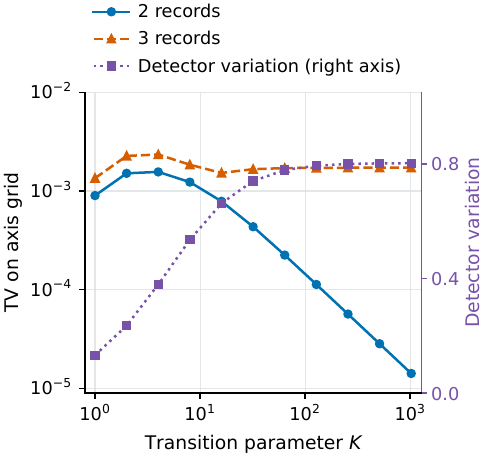}
\caption{Finite-error behavior of the smooth response in Eq.~\eqref{eq:repeated-smooth}, with equal widths $g=0.2$ and $c_0=0.04$. Two-record leakage (blue solid line, circles) decreases as $K$ grows, whereas three-record leakage (orange dashed line, triangles) persists for this family. Both are maximum total-variation distances on the 96-axis grid, read against the left axis. Detector variation $\max f_K-\min f_K$ (purple dotted line, squares; right axis) remains finite and is evaluated on the 65,536-angle grid. These grid extrema are numerical quantities, distinct from the continuum statements proved in the text.}
\label{fig:repeated-finite-error}
\end{figure}

Replacing the hyperbolic tangent in Eq.~\eqref{eq:repeated-smooth} by an antipodally odd sign produces a measurable discontinuous response with constant pair mean $1/2$ and pair variance $c_0^2$. Both its one- and two-readout laws are therefore protected, establishing the necessity of a regularity assumption in Theorem~\ref{thm:two}. Theorem~\ref{thm:three} excludes protection of three independent readouts for this response.

Copying a recorded bit has a different joint law. If every subsequent entry equals the first, only the all-$+$ and all-$-$ strings occur, with probabilities $m$ and $1-m$. A protected one-bit law then remains protected for arbitrarily many copies. This distinction separates repeated access to the unresolved source from repeated storage of an existing record.

\subsection{Numerical results}

For $x=0,a=+1$, consider the exactly protected family
\begin{equation}\label{eq:repeated-sinusoidal}
 f_\epsilon(\lambda)=\frac12+
 \epsilon[1-p(\lambda)]\sin\lambda.
\end{equation}
For $0\le\epsilon\le1/2$, these responses lie in $[0,1]$. Figure~\ref{fig:repeated-access} contrasts a single readout, two independent readouts, and copying the first bit. Equal widths $g\in\{0.15,0.2,0.4,0.8\}$ and strengths $\epsilon\in\{0,0.1,0.25,0.5\}$ are evaluated at 96 equally spaced future axes on $[0,\pi)$. Periodic integration uses 32,768- and 65,536-point angle grids. The smooth family in Eq.~\eqref{eq:repeated-smooth} is evaluated at $g=0.2$, $c_0=0.04$, and $K=2^j$, $j=0,\ldots,10$. Extrema over the axis grid are lower bounds on continuum suprema; Eq.~\eqref{eq:bandwidthbound} supplies a separate analytical continuum bound.

\begin{figure}[tbp]
\centering
\includegraphics[width=\columnwidth]{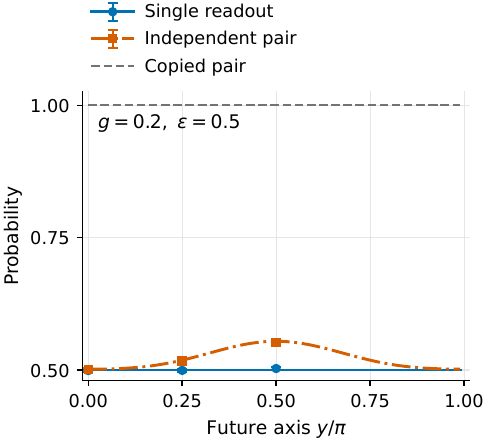}
\caption{Repeated passive readout for the sinusoidal detector in Eq.~\eqref{eq:repeated-sinusoidal}, at $g=0.2$ and $\epsilon=0.5$. The single-readout success probability (blue solid line) is independent of the future axis, whereas the agreement probability of two fresh readouts (orange dash-dotted line) varies. Copies of the first bit agree with certainty (gray dashed line). Lines give model probabilities. Blue circles and orange squares represent the respective 100,000-history Monte Carlo estimates; error bars are pointwise 95\% Wilson intervals.}
\label{fig:repeated-access}
\end{figure}

Continuous Monte Carlo sampling uses 100,000 independent source histories in each of 18 contexts. These comprise $\epsilon=0.25,0.5$ and the smooth responses with $K=1,4,16,64$, each at $y=0,\pi/4,\pi/2$ and $g=0.2$. A wrapped-Cauchy proposal with density $A$ is accepted with probability
\begin{equation}
 \frac{R_B(\lambda;y)}{\coth(\gamma_B)/\pi}.
\end{equation}
The denominator is an analytical maximum of $R_B$, so rejection sampling produces Eq.~\eqref{eq:repeated-source} without discretizing the source angle. Three fresh Bernoulli readouts per accepted angle give 1,800,000 source histories and 5,400,000 detector outcomes. The copying comparison reuses the first bit and contributes no additional independent detector trials.

At $g=0.2,\epsilon=0.5$, each individual readout has success probability $1/2$, whereas the joint probabilities depend on the later axis as shown in Table~\ref{tab:repeated-probabilities}. The two-record total-variation distance between $0$ and $\pi/2$ is $0.0526857392$. The maximum three-record distance on the axis grid is $0.0790286087$. For this sinusoidal response, $d=\epsilon\sin\lambda$ has degree $J=1$, and Eq.~\eqref{eq:bandwidthbound} gives the lower bound $0.0050859086$.

\begin{table}[t]
\centering
\begin{tabular}{@{}lrr@{}}
\toprule
Future axis & $P_y(++)$ & $P_y(R_1=R_2)$\\
\midrule
$0$ & 0.2507726106 & 0.5015452212\\
$\pi/2$ & 0.2771154802 & 0.5542309604\\
\bottomrule
\end{tabular}
\caption{Two independent passive readouts of the sinusoidal detector at $g=0.2$ and $\epsilon=0.5$. The single-readout success probability is $1/2$ at both axes.}
\label{tab:repeated-probabilities}
\end{table}

Figure~\ref{fig:repeated-parameters} displays the dependence on width and strength. Copying the first protected bit gives zero theoretical leakage.

\begin{figure}[tbp]
\centering
\includegraphics[width=\columnwidth]{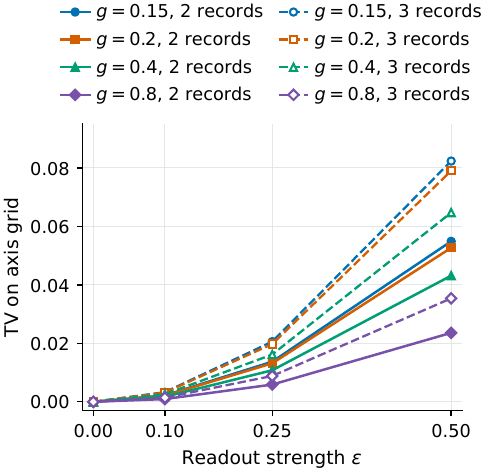}
\caption{Dependence of repeated-readout leakage on the common source width $g$ and detector strength $\epsilon$ for Eq.~\eqref{eq:repeated-sinusoidal}. The maximum pairwise total-variation distance over 96 future axes is shown for two (solid lines, filled markers) and three (dashed lines, open markers) fresh readouts. Widths $g=0.15,0.2,0.4,0.8$ are represented by blue circles, orange squares, green triangles, and purple diamonds, respectively. Grid maxima are lower bounds on the corresponding continuum suprema.}
\label{fig:repeated-parameters}
\end{figure}

In the smooth family shown in Fig.~\ref{fig:repeated-finite-error}, the two-record grid maximum decreases to $1.4166\times10^{-5}$ at $K=1024$, while the angular response range on the integration grid remains $0.8024$. The three-record grid maximum is $0.0017208$ at the same parameters. These results illustrate both the failure of uniform two-record stability and the stronger constraint imposed by a third independent readout.

The Monte Carlo agreement probabilities are compared with deterministic integration in Fig.~\ref{fig:repeated-mc}. Five event probabilities are estimated in each context: one, two, and three consecutive $+$ outcomes, agreement of the first two bits, and agreement of all three. Pointwise uncertainty is described by 95\% Wilson intervals. Exact Clopper--Pearson intervals with Bonferroni correction provide simultaneous 95\% coverage across the 90 event probabilities; all model values lie within those intervals. The largest absolute standardized deviation is approximately $2.98$ for readout events and $3.56$ for the separate source-moment comparisons. The largest change under quadrature refinement is $1.82\times10^{-14}$, a floating-point diagnostic rather than a rigorous discretization bound.

\begin{figure}[tbp]
\centering
\includegraphics[width=\columnwidth]{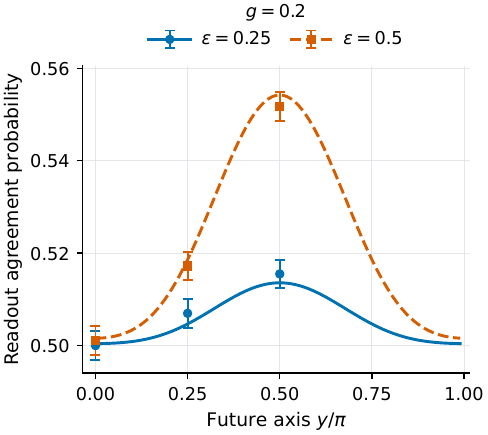}
\caption{Agreement of two independent passive readouts at $g=0.2$ for the sinusoidal response in Eq.~\eqref{eq:repeated-sinusoidal}. Blue solid lines and circles correspond to $\epsilon=0.25$; orange dashed lines and squares correspond to $\epsilon=0.5$. Curves use deterministic integration; markers use 100,000 continuous source histories per context. Error bars are pointwise 95\% Wilson intervals and describe Monte Carlo sampling uncertainty within the specified model. Simultaneous exact intervals for the 90 event probabilities are included in the numerical data.}
\label{fig:repeated-mc}
\end{figure}

The obstruction concerns identical independent readouts of an unchanged hidden angle under the fixed source law. A state update, a modified source measure, shared detector randomness, or restricted access changes those assumptions. The network, controller, and memory constructions in Secs.~\ref{sec:network}--\ref{sec:memory} address different apparatus extensions; they are not constrained by these passive-readout theorems unless they implement the same interface.

\section{Network composition}
\label{sec:exp2}
\label{sec:network}

When independently specified sources meet at a joint detector, a correction
that reproduces their quantum record probabilities need not separate across
the sources. We examine an entanglement-swapping network: source \(S_1\)
supplies \(A,B_1\), source \(S_2\) supplies \(B_2,C\), and a central
detector measures \(B_1B_2\) jointly. Independent-source constraints and
retrocausal descriptions of swapping provide the relevant setting
\cite{branciard,price}. For the positive-history extension defined below,
we characterize the correction rank and determine whether an earlier
archive remains independent of subsequent analyzer choices.

\subsection{Source model and correction rank}

\begin{postulate}[Reference preparations and apparatus]
The source reference measure is
\(\mu_1(\dd\lambda_1)\mu_2(\dd\lambda_2)
=\dd\lambda_1\dd\lambda_2/(2\pi)^2\), with independent angles on the
circle. The emitted pairs have angles
\((\lambda_1,\lambda_1+\pi)\) and
\((\lambda_2,\lambda_2+\pi)\). Outer detectors retain the wrapped-Cauchy
raw factors. The central raw kernel is a Bell effect evaluated on the
incoming hidden product spins. A separate paired quantum description uses
two fixed independent Werner states, calibrated outer effects, and a Bell
measurement. This pairing is model structure, not a consequence of the
isolated source law.
\end{postulate}

Independence in this assumption concerns the preparation factors. The
posterior hidden distribution after imposing the apparatus history need
not factor. In particular, independence after conditioning on a central
outcome is not imposed.

For \(g>0\), define
\begin{align}
 p_g(t)&=\frac{\sinh g}{2\pi(\cosh g-\cos t)},\\
 \theta_{x,a}&=x+\pi\mathbf1_{a=-1},\qquad
 n_x=(\sin x,0,\cos x).
\end{align}
The two outer raw factors, for outcomes \(a,c\in\{+1,-1\}\), are
\begin{align}
 A_{x,a}(\lambda_1)&=p_g(\theta_{x,a}-\lambda_1),\\
 C_{z,c}(\lambda_2)&=p_g(\theta_{z,c}-\lambda_2-\pi).
\end{align}
The Bell effects have the following Pauli signs:
\begin{equation}
\begin{array}{c|rrr}
b&t_{bx}&t_{by}&t_{bz}\\\hline
\Phi^+&1&-1&1\\
\Phi^-&-1&1&1\\
\Psi^+&1&1&-1\\
\Psi^-&-1&-1&-1
\end{array}
\end{equation}
The effects are
\begin{equation}
 \Pi_b=\frac14\left(I+\sum_{j=x,y,z}t_{bj}\sigma_j\otimes\sigma_j\right).
\end{equation}
Writing \(|\lambda\rangle=\cos(\lambda/2)|0\rangle+
\sin(\lambda/2)|1\rangle\) and
\(|\chi\rangle=|\lambda_1+\pi\rangle\otimes|\lambda_2\rangle\),
the central kernel is
\begin{equation}
\begin{aligned}
 q_b(\lambda_1,\lambda_2)&=\langle\chi|\Pi_b|\chi\rangle\\
 &=\frac14-\frac{t_{bx}}4\sin\lambda_1\sin\lambda_2\\
 &\phantom{={}}-\frac{t_{bz}}4\cos\lambda_1\cos\lambda_2.
\end{aligned}
\label{eq:e2-junction}
\end{equation}
It satisfies \(0\le q_b\le1/2\) and
\(\sum_bq_b=1\). The raw branch measure is
\begin{equation}
\begin{aligned}
 R_{abc}^{xz}(\dd\lambda_1\dd\lambda_2)
 &=A_{x,a}(\lambda_1)\\
 &\quad\times q_b(\lambda_1,\lambda_2)C_{z,c}(\lambda_2)\\
 &\quad\times\frac{\dd\lambda_1\dd\lambda_2}{(2\pi)^2},\\
 d_{abc}^{xz}&=\int R_{abc}^{xz}.
\end{aligned}
\end{equation}
The outer factors are raw weights; they are not normalized forward
transition probabilities when their outcomes are summed at fixed angle.

The quantum benchmark uses
\begin{align}
 v&=\sech g,\qquad w=\frac1{1+\tanh^2g},\\
 \omega_w&=\frac14\left(I-w\sum_{j=x,y,z}\sigma_j\otimes\sigma_j\right),\\
 E_a^g(x)&=\frac{I+av\,n_x\cdot\sigma}{2}.
\end{align}
With \(\rho_0=\omega_w^{AB_1}\otimes\omega_w^{B_2C}\), the
complete branch probabilities are
\begin{align}
 G_{abc}^{xz}&=E_a^g(x)\otimes\Pi_b\otimes E_c^g(z),\\
 n_{abc}^{xz}&=\Tr\!\left(\rho_0G_{abc}^{xz}\right).
\end{align}
Both source states and the width are held fixed when future settings or
policies are compared. Quantum network normalization uses the usual
instrument framework \cite{chiribella}.

We ask whether, at a fixed central outcome, the correction
\(C_{ij}=n_{ij}/d_{ij}\) factors into independent functions of the
external record branches, or more generally admits a fixed finite sum
of nonnegative separable terms as the angular resolution increases.

\begin{lemma}[Branch masses]
Let \(r=e^{-2g}\), \(\kappa=w^2v^2\), and
\begin{equation}
 T_b(x,z)=t_{bx}\sin x\sin z+t_{bz}\cos x\cos z.
\end{equation}
Then
\begin{equation}
\begin{aligned}
 d_{abc}^{xz}&=\frac{1+ac\,rT_b(x,z)}{16\pi^2},\\
 n_{abc}^{xz}&=\frac{1+ac\,\kappa T_b(x,z)}{16}.
\end{aligned}
 \label{eq:e2-masses}
\end{equation}
Consequently \(\sum_{a,b,c}d_{abc}^{xz}=\pi^{-2}\) and
\(\sum_{a,b,c}n_{abc}^{xz}=1\).
\end{lemma}
\begin{proof}
The first circular moment is
\begin{equation}
 \int p_g(\theta-\lambda)e^{i\lambda}\dd\lambda
 =e^{-g}e^{i\theta}.
\end{equation}
Applying it to the two factors in
\eqref{eq:e2-junction} gives the raw expression; the two antipodal offsets
cancel in their product. In the quantum trace, each Werner source gives
the correlation tensor \(-wI_3\). Their product contributes \(w^2\),
and the two outer effects contribute \(v^2\). All terms with one
unpaired Pauli operator vanish. Summing the remaining expression over
the outcomes proves the two normalization identities.
\end{proof}

For \(b=\Phi^+\), define the effective branch axes
\(\alpha=x+\pi\mathbf1_{a=-1}\) and
\(\beta=z+\pi\mathbf1_{c=-1}\). Removing the common constant
\(\pi^2\) from the correction gives
\begin{equation}
\begin{aligned}
 D_{ij}&=C_{ij}/\pi^2
       =F\bigl(\cos(\alpha_i-\beta_j)\bigr),\\
 F(t)&=\frac{1+\kappa t}{1+rt}.
\end{aligned}
 \label{eq:e2-correction}
\end{equation}

\begin{theorem}[Nonseparability of the correction]
For every finite \(g>0\), the correction in
\eqref{eq:e2-correction} cannot be written as \(D_{ij}=f_i h_j\)
for all external branches.
\end{theorem}
\begin{proof}
For \(0<r<1\),
\begin{equation}
 \kappa=r\frac{(1+r)^2}{(1+r^2)^2}>r.
\end{equation}
Taking \(\alpha,\beta\in\{0,\pi/2\}\) produces the minor
\begin{equation}
\begin{aligned}
 \Delta_g&=\det\begin{pmatrix}F(1)&1\\1&F(1)\end{pmatrix}\\
 &=\left(\frac{1+\kappa}{1+r}\right)^2-1>0.
\end{aligned}
 \label{eq:e2-minor}
\end{equation}
A product of two scalar branch factors has rank at most one and all
such minors vanish.
\end{proof}

At \(g=0.2\), the minor equals \(0.2807745972218274\). Its
positivity follows from the exact expression in
\eqref{eq:e2-minor}. Figure~\ref{fig:e2-correction} shows the
correction on the eight-branch grid used in the numerical study.

\begin{figure}[tbp]
\centering
\includegraphics[width=\columnwidth]{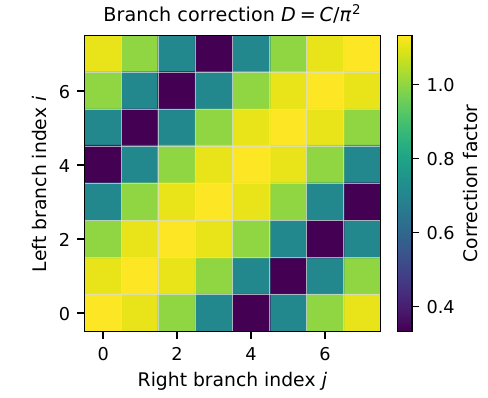}
\caption{Dimensionless network correction \(D=C/\pi^2\) for eight
effective outer branch axes \(\alpha_i=\beta_i=2\pi i/8\), central
outcome \(\Phi^+\), and \(g=0.2\). The entries are deterministic
evaluations of Eq.~\eqref{eq:e2-correction}; the color bar gives the
correction factor and the grid marks individual matrix cells.}
\label{fig:e2-correction}
\end{figure}

\begin{theorem}[Nonnegative rank on even uniform grids]
For any even \(N\ge2\), choose
\(\alpha_i=\beta_i=2\pi i/N\). Then
\begin{equation}
 \operatorname{rank}(D)=\operatorname{rank}_+(D)=N,
\end{equation}
where \(\operatorname{rank}_+\) is the minimum number of terms in
\begin{equation}
 D_{ij}=\sum_{m=1}^{r_{\mathrm{msg}}}u_{im}v_{mj},
 \qquad u_{im},v_{mj}\ge0.
\end{equation}
\end{theorem}
\begin{proof}
Write \(D(\theta)=\kappa/r+(1-\kappa/r)H(\theta)\), where
\begin{align}
 H(\theta)&=\frac1{1+r\cos\theta}\\
 &=\frac1{\sqrt{1-r^2}}
       \sum_{n\in\mathbb Z}(-t)^{|n|}e^{in\theta},\nonumber\\
 t&=\frac{1-\sqrt{1-r^2}}r\in(0,1).
\end{align}
The sampled matrix is circulant. Each nonconstant Fourier eigenvalue,
at \(j=1,\ldots,N-1\), is proportional to
\begin{equation}
 (1-\kappa/r)\sum_{\ell\in\mathbb Z}(-t)^{|j+\ell N|}.
\end{equation}
For even \(N\), every summand has the same sign \((-1)^j\), so the
sum is nonzero. Also \(1-\kappa/r\ne0\). The constant eigenvalue is
the positive row sum. Thus the ordinary rank is \(N\), lower-bounding
the nonnegative rank. The row-selector representation
\(u_{im}=\delta_{im}\), \(v_{mj}=D_{mj}\) supplies the matching
upper bound.
\end{proof}

For \(N=8\), the table has smallest singular value
\(0.15520062120836023\); its singular spectrum is shown in
Fig.~\ref{fig:e2-singular}. This case corresponds to four nominal axes
and both external outcomes, giving eight effective branch axes. The
row-selector upper bound contains a complete external branch label; it
is an algebraic representation, not a physical local-message mechanism.
The theorem is restricted to branch-separable corrections for the
specified raw junction. It gives no bound for arbitrary angle-dependent
corrections, alternative junctions, unrestricted fields, or additional
dynamical variables.

\begin{figure}[tbp]
\centering
\includegraphics[width=\columnwidth]{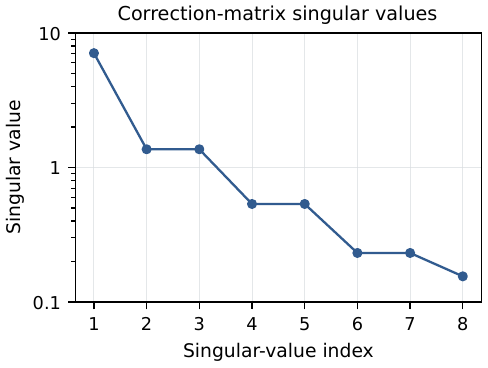}
\caption{Singular values of the eight-branch correction matrix in
Fig.~\ref{fig:e2-correction}. Every value is positive, in agreement
with the full-rank theorem on even uniform grids. The line joins
deterministic evaluations.}
\label{fig:e2-singular}
\end{figure}

\subsection{Archives and adaptive settings}

Add an archived bit \(s=\pm1\) before source \(S_1\) separates. Its
raw factor and paired quantum operation are
\begin{equation}
\begin{aligned}
 L_s(\lambda_1)&=\frac{1+sk\cos2\lambda_1}{2},\\
 M_s&=\sqrt{\frac{I+sk\,Z\otimes Z}{2}},\qquad 0<k<1.
\end{aligned}
 \label{eq:e2-probe}
\end{equation}
The raw measure is now \(L_sR_{abc}^{xz}\); the quantum probability
is evaluated after applying \(M_s\) to \(AB_1\). These are distinct
declared probe extensions. In particular, \(L_s\) is not the Born
response of \(M_s\) on the hidden product spin. Future-policy dependence
is therefore compared within each model, without assuming that their
probe calibrations coincide.

We consider fixed settings \((x,z)=(0,0)\), fixed settings
\((x,z)=(\pi/2,\pi/4)\), and the adaptive choice
\begin{equation}
 (x,z)=
 \begin{cases}
 (0,\pi/4),&s=+1,\\
 (\pi/2,-\pi/4),&s=-1.
 \end{cases}
\end{equation}
All Bell and external outcomes are retained.

\begin{lemma}[Archive marginals]
For fixed settings, the normalized raw model satisfies
\begin{equation}
 P_R(s=+1)=\frac{1+kr\cos2x}{2}.
\end{equation}
Under the stated adaptive policy, \(P_R(s=+1)=1/2\), and the raw
partition changes from \(\pi^{-2}\) to \((1+kr)\pi^{-2}\). The
quantum benchmark gives, under all three policies,
\begin{equation}
 P_Q(s=+1)=\frac{1-kw}{2}.
 \label{eq:e2-archive}
\end{equation}
\end{lemma}
\begin{proof}
Summing the central outcomes removes \(q_b\). Integrating the two
external factors and summing their outcomes leaves the second
wrapped-Cauchy harmonic \(r\cos2x\), yielding the fixed-setting
formula. For the adaptive choice, the two unnormalized archive masses
both equal \((1+kr)/(2\pi^2)\), giving the claimed partition and
equal probabilities. In the quantum model, each selected subsequent
instrument is complete even when its setting depends on \(s\).
Summing its outcomes therefore leaves
\(\Tr(\omega_w M_s^2)=(1-skw)/2\).
\end{proof}

For a full branch \(h=(s,a,b,c)\), the contextual benchmark is
\begin{equation}
 P_Q(h,\dd\lambda_1\dd\lambda_2)
 =n_h\frac{R_h(\dd\lambda_1\dd\lambda_2)}{d_h}.
 \label{eq:e2-embedding}
\end{equation}
Every \(d_h>0\): the outer factors and \(L_s\) are strictly
positive, while each \(q_b\) is positive on a set of positive measure.
Integration of \eqref{eq:e2-embedding} gives \(P_Q(h)=n_h\), so
normalization and \eqref{eq:e2-archive} follow immediately. This
conditional construction embeds the specified quantum probabilities.
It imposes no microscopic mechanism that generates the global factor
\(n_h/d_h\) through local source-field interactions.

The original hidden laws are recovered in a disconnected configuration.
Replace the Bell
junction by the product of the original wrapped-Cauchy analyzers, and
omit the source probe. The apparatus then disconnects, so
\(d_h=d_{h_1}d_{h_2}\) and \(n_h=n_{h_1}n_{h_2}\). For each source,
\begin{align}
 d_{ab}^{xy}&=\frac{\sinh(2g)}{4\pi^2[\cosh(2g)+ab\cos(x-y)]},\\
 n_{ab}^{xy}&=\frac{1-ab\sech(2g)\cos(x-y)}4,\\
 \frac{n_{ab}^{xy}}{d_{ab}^{xy}}&=\frac{2\pi}{Z_{xy}},\\
 Z_{xy}&=\frac{2\sinh(2g)\cosh(2g)}
                   {\pi[\cosh^2(2g)-\cos^2(x-y)]}.
\end{align}
Multiplication by \(2\pi/Z_{xy}\) cancels the reference factor
\(1/(2\pi)\) and recovers the original normalized angle--outcome
law. The disconnected hidden law is therefore its product for the two
sources. No such original-law recovery is claimed at the Bell junction,
whose raw interface is an additional postulate.

Table~\ref{tab:e2-archive} and Fig.~\ref{fig:e2-archive} compare the
archive probabilities under the three future policies.
Numerical sampling used \(g=0.2\), \(k=0.4\), and seed
\(2026091802\). Each of the three policies was sampled under both raw
and contextual laws, with \(100{,}000\) continuous-angle
histories per group: \(600{,}000\) histories in total. Full branch
labels were sampled from the analytic raw masses or direct quantum
matrix traces. Conditional angles were sampled by rejection from the
original positive factors, using
\(q_bL_s\le(1+k)/4\). The sampler generated \(2{,}085{,}587\)
proposals, including unused draws at the end of oversampled batches.

\begin{table}[tbp]
\centering
\footnotesize
\setlength{\tabcolsep}{3pt}
\begin{tabular}{@{}lrr@{}}
\toprule
Policy / law & Exact \(P(s=+)\) & Estimate / 95\% interval\\
\midrule
Fixed Z & 0.634064 & 0.63321\\
Raw & & [0.630218, 0.636192]\\[2pt]
Fixed Z & 0.307499 & 0.30758\\
Contextual & & [0.304727, 0.310448]\\[2pt]
Fixed rotated & 0.365936 & 0.36905\\
Raw & & [0.366064, 0.372046]\\[2pt]
Fixed rotated & 0.307499 & 0.30550\\
Contextual & & [0.302653, 0.308362]\\[2pt]
Adaptive & 0.500000 & 0.50076\\
Raw & & [0.497661, 0.503859]\\[2pt]
Adaptive & 0.307499 & 0.30825\\
Contextual & & [0.305395, 0.311119]\\
\bottomrule
\end{tabular}
\caption{Earlier archived bit, with every central and external outcome
retained. Each two-line entry gives its policy and law, with
\(100{,}000\) histories per entry. Intervals are pointwise 95\% Wilson
intervals. The interval for the raw law with fixed rotated settings
does not contain its exact prediction; no simultaneous coverage is
implied.}
\label{tab:e2-archive}
\end{table}

\begin{figure}[tbp]
\centering
\includegraphics[width=\columnwidth]{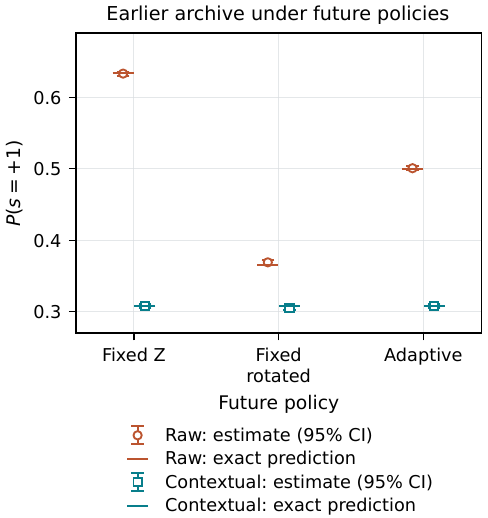}
\caption{Earlier archive probabilities under the three future policies
for \(g=0.2\) and \(k=0.4\). Orange circles and teal squares show
sampled proportions for the raw and contextual laws, respectively;
matching horizontal ticks give their exact predictions. Error bars
are pointwise 95\% Wilson intervals and are narrow on this probability
scale. Each law and policy uses \(100{,}000\) histories. The two laws
are offset horizontally for visibility. All Bell and external outcomes
are retained.}
\label{fig:e2-archive}
\end{figure}

An independent two-dimensional quadrature of the positive raw factors
checked all 32 archive/output branches at unequal axes \(0.37,1.11\).
The largest absolute errors against the harmonic formula were
\(1.77\times10^{-13}\) on a \(128^2\) grid and
\(1.65\times10^{-17}\) on both \(256^2\) and \(512^2\) grids.
Direct \(16\times16\) quantum traces agreed with
\eqref{eq:e2-masses} to \(4.17\times10^{-17}\); full quantum record
normalizations differed from one by at most \(2.23\times10^{-16}\).
These are floating-point diagnostics, not rigorous integration-error
bounds.

To check the conditional-angle sampler, three latent moments were
compared against separate \(512^2\)-node quadrature in every group:
\(\cos2\lambda_1\), \(\sin\lambda_1\sin\lambda_2\), and
\(\cos\lambda_1\cos\lambda_2\). The largest absolute standardized
error was \(1.633\) across these 18 comparisons. The largest
standardized record-cell error was approximately \(3.64\) among 192
reported cells. These correlated pointwise checks are not a simultaneous
coverage statement.

The correction is therefore nonseparable, and its minimum number of
nonnegative separable terms increases with even angular resolution
within the stated ansatz. The raw archive has an exact dependence on
future policy. The contextual history law preserves its archive through
the prescribed quantum normalization, and the numerical results verify
the corresponding sampler. The resulting obstruction and resource
count concern this apparatus extension; a microscopic local realization
and more general network corrections remain outside these claims.

\section{Readable controllers}\label{sec:controller}
A controller can restore the distribution of an earlier outcome by supplying the inverse of a continuation weight. Does this cancellation also protect a record of the controller itself? We examine this question in a Gaussian boundary ensemble whose branch masses can be evaluated exactly. The distinction between an invariant outcome distribution and an invariant joint record leads to explicit access and robustness constraints.

\subsection{Boundary ensemble and joint records}
Let $h\in\{0,1\}$ denote an earlier archived branch and $u\in\{0,1\}$ a later choice between two policies. The continuation weights are
\begin{equation}
\begin{aligned}
 (R_{0,0},R_{1,0})&=(1,\rho),\\
 (R_{0,1},R_{1,1})&=(\rho,1),
 \qquad \rho>1.
\end{aligned}
\end{equation}
Each policy may use the archived value $h$. The numerical benchmark takes $\rho=\coth^2(0.4)$.

The controller has dimensionless coordinates $Q=(Q_1,Q_2)$ and radial coordinate $T=|Q|^2/2$. Its stiffness is $r=R^\alpha$, its bare potential offset is $\zeta\log R$, and an optional radial cutoff restricts $T\le B$. Define the unnormalized history measure by
\begin{equation}\label{eq:controllerraw}
\begin{aligned}
 M_u(h,\dd^2Q)
 &=\frac12 R_{h,u}^{1-\zeta}
   e^{-R_{h,u}^{\alpha}T}\\[-1pt]
 &\hspace{1em}\times\mathbf1_{T\le B}
   \frac{\dd^2Q}{2\pi}.
\end{aligned}
\end{equation}
One global normalization over $h$ and $Q$ determines the probability law. The boundary ensemble and the stiffness interface $r=R^\alpha$ are model assumptions; they do not specify microscopic source--controller dynamics.

The record under consideration is $(h,Q)$, assumed jointly accessible before the later policy is selected. This timing is essential: a controller readout obtained only after that selection need not have a policy-independent distribution. The principal case is $(\alpha,\zeta,B)=(1,0,\infty)$.

\begin{lemma}[Raw and normalized Gaussian factors]
The branch mass obtained from \eqref{eq:controllerraw} is
\begin{equation}
 m_u(h)=\frac12 R_{h,u}^{1-\zeta-\alpha}
 \bigl(1-e^{-R_{h,u}^{\alpha}B}\bigr).
\end{equation}
For $B=\infty$, the normalized conditional law is
\begin{equation}
 Q\mid h,u\sim N\!\left(0,R_{h,u}^{-\alpha}I_2\right).
\end{equation}
A normalized conditional thermal kernel integrates to one and therefore supplies no inverse-stiffness factor to the branch mass.
\end{lemma}
\begin{proof}
After integrating over the polar angle, the measure $\dd^2Q/(2\pi)$ becomes $\dd T$. Hence
\begin{equation}
 \int_0^B e^{-rT}\dd T=\frac{1-e^{-rB}}{r}.
\end{equation}
Multiplication by the remaining factors gives $m_u(h)$. Conditional normalization removes this integral from the controller law.
\end{proof}

For the principal case, the Gaussian integral is $1/R_{h,u}$, so each corrected branch has probability $1/2$. Preservation of this marginal is insufficient for preservation of the complete record: at fixed supported $h$, the normalized conditional law of $Q$ must also be independent of $u$, as in Theorem~\ref{thm:conditional}.

\begin{theorem}[Protected marginals and an unprotected joint record]\label{thm:controller}
At $(\alpha,\zeta,B)=(1,0,\infty)$, both $P_u(h)=1/2$ and the marginal law of $Q$ are independent of $u$. Their joint law nevertheless satisfies
\begin{equation}
\begin{aligned}
 D&:=\TV\!\left(P_0^{H,Q},P_1^{H,Q}\right)\\
  &=(\rho-1)\rho^{-\rho/(\rho-1)}>0.
\end{aligned}
\end{equation}
The binary radial readout $\mathbf1_{T>T_*}$, together with $h$, retains this entire distinction, where
\begin{equation}
 T_* =\frac{\log\rho}{\rho-1}.
\end{equation}
For equal prior probabilities of the two policies, the optimal discrimination success is $(1+D)/2$.
\end{theorem}
\begin{proof}
The two policies interchange the conditional Gaussian laws while leaving their equal mixture, the marginal law of $Q$, unchanged. Conditional on $h$, the angular coordinate is uniform under both policies. Their distinction therefore reduces to the radial densities $e^{-T}$ and $\rho e^{-\rho T}$, which cross at $T_*$. Integration gives
\begin{equation}
\begin{aligned}
 \TV\!\left(\operatorname{Exp}(1),
                  \operatorname{Exp}(\rho)\right)
 &=e^{-T_*}-e^{-\rho T_*}\\
 &=D.
\end{aligned}
\end{equation}
The equal branch probabilities give the same total variation for the joint record. The optimal event for identifying policy zero is
\begin{equation}
 \{h=0,T>T_*\}\ \cup\ \{h=1,T\le T_*\}.
\end{equation}
Binning at the density crossing preserves the sign of the density difference and hence its total variation.
\end{proof}

The access condition is substantive. An angular-only readout, such as $\operatorname{sgn}Q_1$, is a protected alternative: isotropy gives the same fair sign conditional on every $h,u$. Thus informative radial access can reveal the policy through correlations with $h$, even though neither separate marginal reveals it. Figure~\ref{fig:controller-radial} shows the conditional radial densities and their crossing.

\begin{figure}[tbp]
\centering
\includegraphics[width=\columnwidth]{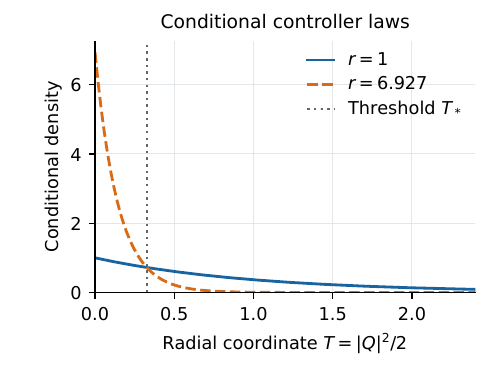}
\caption{Conditional radial densities of the Gaussian controller in the principal ensemble, with $\rho=\coth^2(0.4)$. The solid blue and dashed orange curves have rates $1$ and $\rho$, respectively; the dotted gray line marks their crossing at $T_*=\log\rho/(\rho-1)$. This threshold defines an optimal binary radial readout. Together with the branch record $h$, this readout retains the full joint-record total variation $D$, despite policy-independent separate marginals.}
\label{fig:controller-radial}
\end{figure}

\subsection{Noise, distribution changes, and robustness}
\begin{proposition}[Finite Gaussian readout noise]
Let the accessible coordinate be $Z=Q+\eta$, where $\eta\sim N(0,sI_2)$ is independent of $Q,h,u$, and retain the principal raw ensemble. Define
\begin{equation}
\begin{aligned}
 k(s)&=\frac{1+s}{\rho^{-1}+s},\\
 G(k)&=(k-1)k^{-k/(k-1)},\qquad k>1,
\end{aligned}
\end{equation}
with $G(1)=0$. The joint-record total variation is $D(s)=G(k(s))$. For $\rho>1$, it is positive at every finite $s$. For $0<\delta<D(0)$, the minimum variance within this readout family is
\begin{equation}
 s_{\min}=\frac{1-k_\delta/\rho}{k_\delta-1},
 \qquad G(k_\delta)=\delta,
\end{equation}
where $1<k_\delta<\rho$.
\end{proposition}
\begin{proof}
The conditional covariance scalars become $1+s$ and $\rho^{-1}+s$. A common rescaling reduces their total variation to Theorem~\ref{thm:controller} with variance ratio $k(s)$. For $k>1$,
\begin{equation}
 \frac{\dd\log G}{\dd k}
 =\frac{\log k}{(k-1)^2}>0.
\end{equation}
The ratio $k(s)$ decreases continuously from $\rho$ to one as $s\to\infty$. Monotonicity gives the unique variance threshold and positivity at every finite $s$.
\end{proof}
This bound concerns a coordinate-resolved additive Gaussian channel. It does not apply to angular-only readout, nor does it assign a disturbance or energetic cost to a physical measurement. The raw controller ensemble is unchanged by the stipulated readout noise. The remaining joint-record distinction is shown in Fig.~\ref{fig:controller-noise}.

\begin{figure}[tbp]
\centering
\includegraphics[width=\columnwidth]{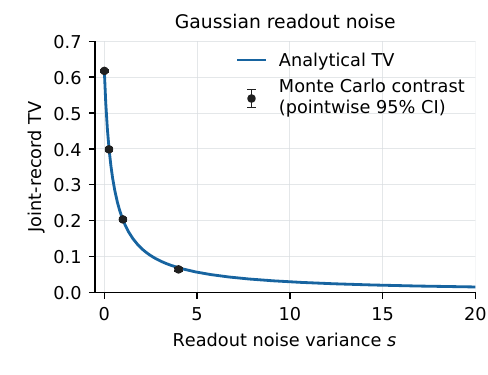}
\caption{Joint-record total variation under coordinate-resolved additive Gaussian readout noise of variance $s$ per coordinate. The solid blue curve is the exact prediction $D(s)=G(k(s))$. Black circles show the Monte Carlo difference of probabilities for the analytically fixed optimal classification event; error bars give pointwise approximate $95\%$ intervals. These event contrasts estimate $D(s)$ without fitting a classifier to the samples. Every finite variance leaves a positive distinction in this readout family.}
\label{fig:controller-noise}
\end{figure}

\begin{proposition}[Unrestricted conditional-law repair]
Retain equal branch probabilities and the swapped-law construction. If each of the two conditional controller laws may change by at most $\epsilon$ in total variation, the least residual joint-record leakage is
\begin{equation}
 \max(0,D-2\epsilon).
\end{equation}
Exact protection requires and permits $\epsilon=D/2$.
\end{proposition}
\begin{proof}
The triangle inequality bounds the separation of the replacement laws from below by $D-2\epsilon$. For $0\le t\le1/2$, replace each law by its mixture with weight $t$ of the other law. Each changes by $tD$, while their separation becomes $(1-2t)D$. At $t=1/2$, both equal their common midpoint.
\end{proof}
The attainable replacements need not remain Gaussian. The proposition specifies the minimum change of conditional distributions without selecting an interaction that produces it. Figure~\ref{fig:controller-law-change} displays the attainable tradeoff up to exact protection.

\begin{figure}[tbp]
\centering
\includegraphics[width=\columnwidth]{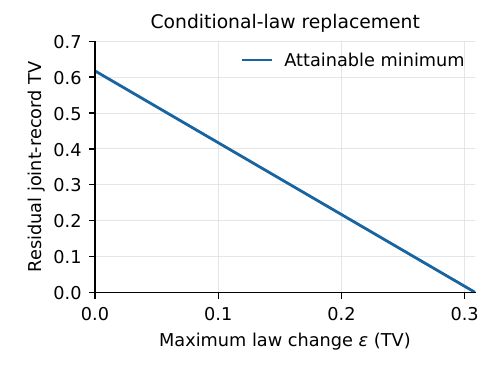}
\caption{Sharp minimum residual joint-record total variation when each conditional controller law can change by at most $\epsilon$ in total variation. Mixing the two original laws attains $D-2\epsilon$ for $0\le\epsilon\le D/2$; their common midpoint gives exact protection at $D/2=0.3086349355$. The replacements are unrestricted and need not be Gaussian.}
\label{fig:controller-law-change}
\end{figure}

\begin{theorem}[Stiffness and bare-potential classification]
For $\rho>1$, $B=\infty$, and exact coordinate readout, branch protection in \eqref{eq:controllerraw} holds if and only if $\alpha+\zeta=1$. Full $(h,Q)$ protection holds if and only if $(\alpha,\zeta)=(0,1)$.
\end{theorem}
\begin{proof}
The branch mass ratio is $\rho^{1-\zeta-\alpha}$, which equals one precisely when $\alpha+\zeta=1$. At equal branch masses, complete-record protection requires equal conditional Gaussian covariances. Since $\rho>1$, this forces $\alpha=0$ and hence $\zeta=1$. Conversely, these values remove every dependence on $R$ from the raw measure.
\end{proof}
The stiffness-only family $\zeta=0$ therefore contains no exactly protected full-record model. The protected point instead supplies the bare factor $R^{-1}$ with constant stiffness. This identifies the required potential offset but leaves its physical origin unspecified. Exact protection occupies no open neighborhood in this two-parameter family; the conclusion concerns the independent parameter variations admitted here. Figure~\ref{fig:controller-stiffness} compares branch and joint-record total variation along the stiffness-only family.

\begin{figure}[tbp]
\centering
\includegraphics[width=\columnwidth]{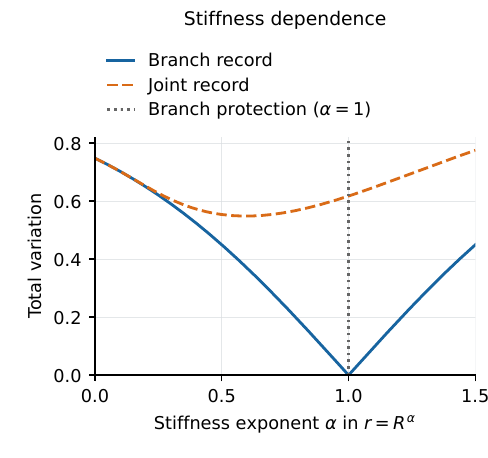}
\caption{Analytical branch (solid blue) and joint-record (dashed orange) total variation as the stiffness exponent $\alpha$ varies with zero bare offset ($\zeta=0$) and no radial cutoff. At $\alpha=1$ (dotted gray line), the branch distribution is protected while the full record remains policy dependent. Full-record protection requires the distinct parameter choice $(\alpha,\zeta)=(0,1)$.}
\label{fig:controller-stiffness}
\end{figure}

A finite radial range also changes the branch normalization. At $\alpha=1$ and $\zeta=0$, the two branch masses are proportional to $1-e^{-B}$ and $1-e^{-\rho B}$. The resulting branch total variation is
\begin{equation}
 \frac{e^{-B}-e^{-\rho B}}
      {2-e^{-B}-e^{-\rho B}}>0,
 \qquad 0<B<\infty.
\end{equation}
It vanishes as $B\to\infty$, while the joint-record total variation approaches the nonzero value in Theorem~\ref{thm:controller}.

The Monte Carlo calculations use $2{,}300{,}000$ controller histories across eleven configurations and both policies. The principal configuration has $150{,}000$ histories per policy; each remaining configuration has $100{,}000$. Conditional radii are sampled continuously by inverse transform, with uniform polar angle and independent random streams for the two policies. The retained variables are $h$, the controller coordinates and radius, and the observed coordinates when readout noise is present. A further $553$ deterministic parameter values cover stiffness, radial range, readout noise, and bare potential.

For $\rho=6.9270683784$, the exact joint total variation is $D=0.6172698709$, while both separate marginal total variations are zero. The sampled difference of probabilities for the analytically fixed optimal event is $0.6173733\pm0.0028150$, where the uncertainty is the pointwise approximate $95\%$ interval half-width. The threshold is $T_*=0.3265419878$, the equal-prior discrimination success is $0.8086349355$, and the minimum unrestricted conditional-law change for exact protection is $D/2=0.3086349355$.

Reducing the joint total variation to $0.01$ using passive Gaussian readout noise requires $s\ge30.90598036$ per coordinate in the stated dimensionless units. For a finite cutoff $B=1$, the branch and joint total variations are $0.2249338741$ and $0.5319232804$, respectively. Replacing the raw Gaussian factor by a normalized thermal conditional preparation leaves the branch weights proportional to $R$; the corresponding branch and joint total variations are $0.7476999182$ and $0.8600509558$.

These results distinguish a controller's scalar contribution to branch normalization from the information retained in its accessible coordinates. A source--controller dynamics that realizes the correction must also determine which controller observables can be read before later choices and how such readout changes the joint history law.

\section{Variable widths and memory}
\label{sec:memory}
\label{sec:experiment4}

The access and composition constraints in Sections~\ref{sec:repeated}
and~\ref{sec:network} concern records produced by specified apparatus rules.
An operational completion must also specify how those rules respond to changes
in propagation width and to apparatus reuse. We consider the source family
\begin{equation}
 \begin{aligned}
 &P_0(a,b\mid x,y,g,h)\\
 &\qquad=\frac{1-ab\sech(g+h)\cos(x-y)}4,
 \end{aligned}
 \label{eq:e4-target}
\end{equation}
where $g,h>0$ are independently variable local widths and $a,b\in\{-1,+1\}$.
A fixed singlet preparation with correlated classical noise reproduces this
family and defines a temporal witness of controller retention. Its resource
requirements lead to a conditional obstruction for finite total shared
dimension. This construction complements the controller analysis in
Section~\ref{sec:controller} at the operational level: its quantum instrument
is specified explicitly, while a derivation from the original retrocausal
hidden-angle dynamics remains outside its scope.

\subsection{Operational completion and reset}

Write $n_x=(\sin x,0,\cos x)$ and
$\Pi_a(x)=(I+a n_x\cdot\sigma)/2$ for the nominal spin projector.

\begin{postulate}[Shared controller and local re-preparation]
\label{post:e4-model}
Each independent trial prepares a singlet and a classical variable $\Omega$,
shared between the wings and sampled independently of the externally chosen
source axes, widths, and adaptive policy, with density
\begin{equation}
 p(\omega)=\frac{1}{2\cosh(\pi\omega/2)}.
 \label{eq:e4-density}
\end{equation}
Alice measures projectively at $x+g\Omega$, Bob at $y-h\Omega$.
Following its outcome, each detector re-prepares the corresponding
\emph{nominal} eigenstate, $\Pi_a(x)$ or $\Pi_b(y)$.
In the fresh extension, subsequent local uses draw independent controllers,
independent also of the incoming spin and complete earlier record. In the
retained extension, the controller is reused throughout the local sequence.
Later controls may depend on archived outcomes and hence correlate with the retained controller.
Independent source trials always refresh the source and its initial controller.
\end{postulate}

The initial paired use has a common controller in both extensions; independent
local refreshes begin afterwards. The continuous controller is an explicit
shared resource, represented by a hybrid classical--quantum ensemble. The
operational-completion claim is that fresh local resets reproduce both
\eqref{eq:e4-target} and the nominal measure-and-prepare continuation instrument
on every finite adaptive history formed from these operations.

\begin{lemma}[Controller transform]
\label{lem:e4-transform}
The distribution~\eqref{eq:e4-density} has
\begin{equation}
 \begin{aligned}
 F(\omega)&=\frac2\pi\arctan e^{\pi\omega/2},\\
 \E e^{it\Omega}&=\sech t.
 \end{aligned}
 \label{eq:e4-transform}
\end{equation}
It can be sampled by
$\Omega=(2/\pi)\log\tan(\pi U/2)$ with $U$ uniform on $(0,1)$.
\end{lemma}
\begin{proof}
Differentiating $F$ gives $p$, and $F$ has limits zero and one. In the Fourier
integral, substitute $z=e^{\pi\omega/2}$ to obtain
\[
 \begin{aligned}
 &\frac2\pi\int_0^\infty\frac{z^{2it/\pi}}{1+z^2}\dd z\\
 &\quad=\frac1{\sin[\pi(1+2it/\pi)/2]}=\sech t.
 \end{aligned}
\]
The integral identity follows from the beta integral, whose parameter here has
real part one. Inversion of $F$ gives the sampler.
\end{proof}

\begin{theorem}[Variable-width completion]
\label{thm:e4-completion}
Assumption~\ref{post:e4-model} reproduces~\eqref{eq:e4-target} for all positive
$g,h$ and planar axes. An independently prepared input $\rho$, uncorrelated with
its controller, experiences
\begin{equation}
 \begin{aligned}
 \mathcal I_a^{g,x}(\rho)&=\Tr[E_a^g(x)\rho]\Pi_a(x),\\
 E_a^g(x)&=\frac{I+a\sech g\,n_x\cdot\sigma}{2}.
 \end{aligned}
 \label{eq:e4-instrument}
\end{equation}
With fresh local controllers, the same instrument describes every subsequent
branch of a finite adaptive protocol composed of these local operations and
classical outcome-dependent choices after the paired source use.
\end{theorem}
\begin{proof}
At fixed $\omega$, the singlet correlation is
$-\cos[x-y+(g+h)\omega]$, with uniform single-wing marginals. Averaging by
Lemma~\ref{lem:e4-transform} proves~\eqref{eq:e4-target}. Averaging a shifted
projector on an independent input removes its sine term and multiplies its
cosine term by $\sech g$. Nominal re-preparation therefore gives
\eqref{eq:e4-instrument}. Freshness makes this same conditional map applicable
after every earlier history, even when that history selects the next width and
axis. Induction on the finite protocol tree proves the joint-history statement.
The outcome-summed maps are trace preserving; summing complete future branches
therefore preserves earlier operational records.
\end{proof}

Independent wing controllers at the \emph{source} instead produce visibility
$\sech g\sech h$. Since
\[
 \sech(g+h)=\frac{\sech g\sech h}{1+\tanh g\tanh h},
\]
this change is observable for positive widths. Shared noise is sufficient for
this construction, but the construction does not establish its necessity or
uniqueness among operational completions. Record preservation follows from the
trace-preserving quantum instruments. Figure~\ref{fig:e4-source} compares the
two source visibilities with their sampled estimates.

\begin{figure}[tbp]
 \centering
 \includegraphics[width=\columnwidth]{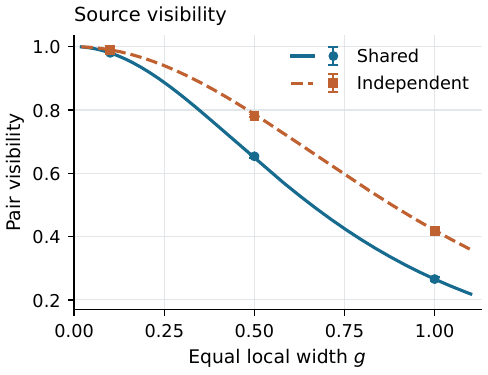}
 \caption{Equal-width source visibility for shared and independent wing
 controllers. The solid blue curve and circles represent a shared controller;
 the dashed orange curve and squares represent independent controllers.
 Curves give analytic predictions, and markers show sampled outcomes with
 pointwise 95\% Wilson intervals, sometimes smaller than their markers. The independent-controller visibility differs from the target
 law~\eqref{eq:e4-target}.}
 \label{fig:e4-source}
\end{figure}

Equality of the one-use prepared-input maps does not imply equality under
reuse. To distinguish the two extensions, consider two nominal $Z$ measurements
starting in $|0\rangle$. The null hypothesis is equality of the final plus
probabilities under fresh and retained controllers.

\begin{theorem}[Reset-sensitive endpoint probability]
\label{thm:e4-reset}
Starting from $|0\rangle$, with nominal $Z$ re-preparation after each outcome,
\begin{equation}
 P_{\rm fresh}(s_2=+)=\frac{1+\sech g\sech h}{2},
 \label{eq:e4-fresh}
\end{equation}
\begin{equation}
 \begin{aligned}
 P_{\rm retained}(s_2=+)
 &=\frac12+\frac14\bigl[\sech(g-h)\\
 &\hspace{3.3em}+\sech(g+h)\bigr].
 \end{aligned}
 \label{eq:e4-retained}
\end{equation}
\end{theorem}
\begin{proof}
The first outcome has conditional mean $\cos(g\Omega_1)$, and the second has
conditional mean $s_1\cos(h\Omega_2)$. Consequently
$\E s_2=\E[\cos(g\Omega_1)\cos(h\Omega_2)]$.
Independence gives~\eqref{eq:e4-fresh}; a common controller gives
\eqref{eq:e4-retained} by the cosine product identity and
Lemma~\ref{lem:e4-transform}.
\end{proof}

At $g=h=0.5$ the theoretical difference is $0.018789701933$.
Thus the two extensions have distinct temporal predictions despite equality of
their one-use prepared-input maps. Simulations comprised 200,000 independent
two-stage trials per model; Table~\ref{tab:e4-reset} gives the final-outcome
counts, and Figure~\ref{fig:e4-reset} displays the comparison. The estimated
retained-minus-fresh difference is $0.018595$, with independent-sample
normal 95\% interval $[0.016762,0.020428]$. These synthetic data quantify the
distinguishability of the specified mechanisms; they are not measurements of a
physical detector.

\begin{table}[tbp]
\centering
\small
\begin{tabular}{@{}lrr@{}}
\toprule
Quantity & Fresh & Retained\\
\midrule
Plus count & $178750$ & $182469$\\
Independent trials & $200000$ & $200000$\\
Estimate & $0.893750$ & $0.912345$\\
95\% Wilson lower & $0.892392$ & $0.911098$\\
95\% Wilson upper & $0.895093$ & $0.913576$\\
Theory & $0.8932238665$ & $0.9120135684$\\
\bottomrule
\end{tabular}
\caption{Two-stage reset comparison at $g=h=0.5$, starting in $|0\rangle$.
The intervals are pointwise 95\% Wilson intervals from independent simulated
trials.}
\label{tab:e4-reset}
\end{table}

\begin{figure}[tbp]
 \centering
 \includegraphics[width=\columnwidth]{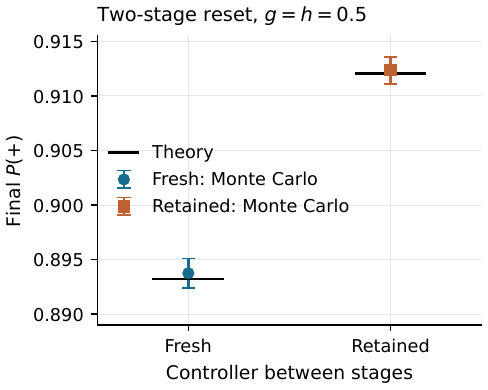}
 \caption{Two-stage reset witness at $g=h=0.5$, starting in $|0\rangle$.
 Blue circles (fresh controller) and orange squares (retained controller),
 with pointwise 95\% Wilson intervals, summarize 200,000 independent
 simulated trials per prescription. Black horizontal bars give the
 exact predictions~\eqref{eq:e4-fresh} and~\eqref{eq:e4-retained}.}
 \label{fig:e4-reset}
\end{figure}

For a retained controller and a specified history $s_1,\ldots,s_T$, let the
nominal axes be $x_j=x_j(s_1,\ldots,s_{j-1})$, with $s_0=+1$ and initial
state $\Pi_+(x_0)$. Writing $d_j=x_j-x_{j-1}$, the conditional history weight is
\begin{equation}
 \prod_{j=1}^T
 \frac{1+s_{j-1}s_j\cos(d_j+g_j\Omega)}2.
 \label{eq:e4-history}
\end{equation}
Expanding each cosine into exponentials and replacing $e^{it\Omega}$ by
$\sech t$ integrates this history exactly. For fresh controllers, its average
instead factorizes, replacing each cosine by
$\sech g_j\cos(x_j-x_{j-1})$.

The adaptive comparison includes every three-stage deterministic policy with widths
$(0.25,0.5,0.75)$, first axis zero, second axis chosen from $\{0,\pi/3\}$ by the
first outcome, and third axis from $\{0,\pi/2\}$ by the first two outcomes.
There are $2^{2+4}=64$ policies and eight histories each. Complete-history total
variation between fresh and retained models ranges from $0.00526175$ to
$0.07597691$, as shown in Figure~\ref{fig:e4-adaptive}. Eight fixed
representative policies were also sampled with
100,000 independent trials per policy and model. Fresh analytic probabilities
agreed with an independent density-matrix implementation
of~\eqref{eq:e4-instrument} to $1.11\times10^{-16}$. Retained Fourier weights
agreed with ordinary numerical integration for one complete policy to the same
reported residual. These are floating-point consistency diagnostics, not
rigorous integration-error bounds.

\begin{figure}[tbp]
 \centering
 \includegraphics[width=\columnwidth]{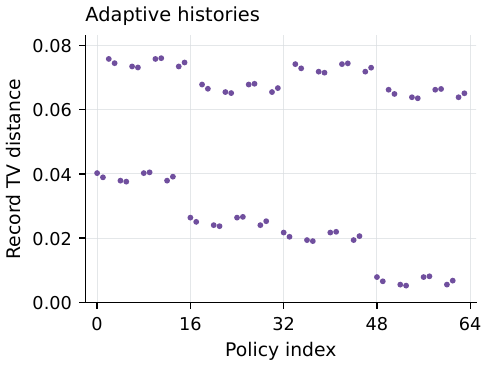}
 \caption{Exact complete-history total variation between fresh and retained
 controllers for all 64 specified three-stage adaptive policies. Each point
 compares the two eight-history distributions for one policy; it is an
 analytic evaluation, with no sampling error bars.}
 \label{fig:e4-adaptive}
\end{figure}

Interventions and full temporal records are central to operational descriptions
of quantum memory~\cite{pollock,taranto}. The endpoint witness above resolves
two reset prescriptions; it neither certifies a general environment dimension
nor implies backwards signalling.

\subsection{Finite shared resources}

\begin{postulate}[Counted resource and width-local access]
\label{post:e4-resource}
A bipartite preparation $\sigma_{AB}$ is fixed independently of both widths.
Each observable depends on its own width and local controls only. All correlated
controller and environment registers count towards the shared resource.
Initially independent local ancillas and private randomness may be absorbed into
effective local observables. Communication and width-dependent joint
preparations are excluded from this resource statement.
\end{postulate}

\begin{theorem}[Variable-width rank]
\label{thm:e4-rank}
For distinct positive $g_1,\ldots,g_N$ and distinct positive
$h_1,\ldots,h_N$, the matrix $K_{ij}=\sech(g_i+h_j)$ has rank $N$.
Under Assumption~\ref{post:e4-resource}, correlations generated from a fixed
finite-dimensional total state with local dimensions $d_A,d_B$ have matrix rank
at most $\min(d_A^2,d_B^2)$. No fixed finite total dimension therefore reproduces
\eqref{eq:e4-target} on the entire continuum of widths.
\end{theorem}
\begin{proof}
Put $u_i=e^{2g_i}$ and $v_j=e^{-2h_j}$. Then
\[
 K_{ij}=2e^{g_i-h_j}(u_i+v_j)^{-1}.
\]
The diagonal prefactors are nonzero and the Cauchy determinant satisfies
\[
 \begin{aligned}
 &\det[(u_i+v_j)^{-1}]\\
 &\quad=\frac{\prod_{i<j}(u_j-u_i)(v_j-v_i)}
                  {\prod_{i,j}(u_i+v_j)}\ne0.
 \end{aligned}
\]
For the resource bound, write an operator Schmidt decomposition
$\sigma_{AB}=\sum_{k=1}^r c_kX_k\otimes Y_k$, with
$r\leq\min(d_A^2,d_B^2)$. The correlation matrix
$-\Tr[\sigma_{AB}(A(g_i)\otimes B(h_j))]$ is a sum of $r$ outer products.
Its rank is therefore at most $r$.
\end{proof}

This counts the classical controller, and does not require an unbounded
\emph{spin} dimension: Theorem~\ref{thm:e4-completion} uses qubits and a continuous
classical resource. Treating unlimited shared randomness as free invalidates
the total-dimension conclusion. Such resource-accounting distinctions are
essential in dimension-witness work~\cite{sikora}. For a purely classical shared
label with $M$ values and local product responses, the rank is at most $M$.
For this particular singlet construction with $M$ discrete shifted-axis labels,
cosine addition gives rank at most $2M$.

\begin{corollary}[Necessary approximation error]
\label{cor:e4-rank-error}
If $\widetilde K$ has rank at most $r<N$ and
$\max_{ij}|K_{ij}-\widetilde K_{ij}|\leq\epsilon$, then
\begin{equation}
 \epsilon\geq\sigma_{r+1}(K)/N.
 \label{eq:e4-rank-error}
\end{equation}
\end{corollary}
\begin{proof}
The spectral error of any rank-$r$ approximation is at least
$\sigma_{r+1}(K)$, while the Frobenius bound gives
\[
 \|K-\widetilde K\|_2\leq\|K-\widetilde K\|_F\leq N\epsilon.
\]
\end{proof}

This necessary bound need not be attainable by physical instruments. For
12 equally spaced widths on $[0.05,1.5]$, the kernel has exact rank 12 and
numerical rank nine under the standard double-precision tolerance. The rapidly
decreasing singular values make robust inference substantially weaker than the
exact-rank statement. Singular values below numerical resolution cannot certify
a dimension bound.

The exact obstruction nevertheless permits finite positive approximations on
bounded width intervals. Such approximations can be implemented by the same
qubit instruments with a discrete shared controller.

\begin{theorem}[Uniform finite-controller approximation]
\label{thm:e4-quadrature}
Fix $0\leq t\leq T$. Partition $[-L,L]$ into bins of width $\Delta$ and place each
bin's probability mass from~\eqref{eq:e4-density} at its midpoint. Place the two
tail masses at $-L$ and $+L$. The resulting positive normalized symmetric
distribution has cosine transform $\phi_M$ satisfying
\begin{equation}
 \begin{aligned}
 \sup_{0\leq t\leq T}|\phi_M(t)-\sech t|
 &\leq\min\{2,B_1,B_2\},\\
 B_1&=\frac{T\Delta}{2}+2\delta_L,\\
 B_2&=\frac{(T+T^2)\Delta^2}{8}+2\delta_L,
 \end{aligned}
 \label{eq:e4-finite-bound}
\end{equation}
where $\delta_L=(4/\pi)\arctan e^{-\pi L/2}$ is the total tail mass.
At fixed $T$, an attainable label count is
$M=O(\epsilon^{-1/2}\log(1/\epsilon))$ for uniform error $\epsilon$.
\end{theorem}
\begin{proof}
Each interior displacement is at most $\Delta/2$, and the cosine derivative has
magnitude at most $T$, giving the linear bound. Tail discrepancies contribute
at most $2\delta_L$. For the quadratic bound, write $u=\omega-m$ in a bin with
midpoint $m$ and half-width $h=\Delta/2$. Integration by parts yields
\[
 \begin{aligned}
 &\int_{-h}^h u p(m+u)\dd u\\
 &\quad=\frac12\int_{-h}^h(h^2-u^2)p'(m+u)\dd u.
 \end{aligned}
\]
Thus the sum of absolute signed bin moments is at most
$\Delta^2\int_{\mathbb R}|p'|/8=\Delta^2/8$, since the symmetric unimodal density
has $\int|p'|=2p(0)=1$. Taylor expansion of the cosine at each midpoint bounds
the summed linear terms by $T\Delta^2/8$ and the quadratic remainders by
$T^2\Delta^2/8$. Both cosine transforms lie between minus one and one, giving
the universal bound two. Choosing $L=O(\log(1/\epsilon))$ and
$\Delta=O(\sqrt\epsilon)$ proves the stated constructive scaling.
\end{proof}

This is an upper bound for an explicit positive approximation, without an
optimality claim. For the source law, visibility error $\epsilon$ gives per-outcome
probability error at most $\epsilon/4$ and total variation at most $\epsilon/2$.
Mathematical bounds assume exact bin masses; evaluated CDF differences incur
ordinary floating-point rounding.

Finite-controller approximations were evaluated with 8 through 1024 labels,
$L=8$, and 4001 points on $0\leq t\leq3$. At 1024 labels the maximum grid error was
$1.43921\times10^{-5}$, whereas the proven continuous-domain upper bound was
$3.76526\times10^{-4}$. Figure~\ref{fig:e4-finite} distinguishes the evaluated errors
from the uniform analytic bound.

\begin{figure}[tbp]
 \centering
 \includegraphics[width=\columnwidth]{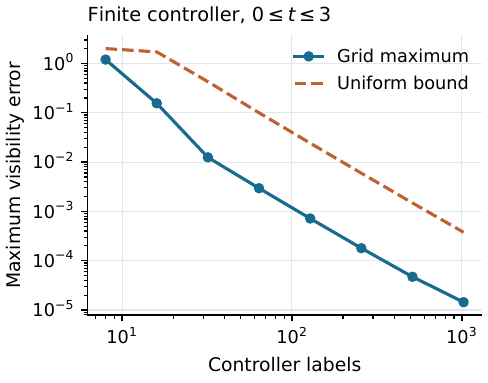}
 \caption{Positive finite-controller approximation on $0\leq t\leq3$.
 Blue circles joined by a solid line show the grid maximum evaluated at
 4001 points using $L=8$. The orange dashed curve is the continuous-domain bound~\eqref{eq:e4-finite-bound}; it controls all
 widths in the stated interval rather than only the evaluated grid.}
 \label{fig:e4-finite}
\end{figure}

The simulations comprise 3,600,000 independent trials and
8,800,000 detector outcomes. Source sampling comprised 16 cells of 100,000
trials: width pairs $(0.1,0.1)$, $(0.2,0.7)$, $(0.5,0.5)$ and $(1,1)$,
axis differences $0$ and $\pi/3$, and shared or independent wing controllers.
The largest absolute binomial standardized deviation from its analytic source
probability was 2.34421; for the sampled adaptive histories it was 2.63793.
The pointwise confidence intervals are descriptive and do not assert
simultaneous coverage over the parameter and policy comparisons. The
computational data include sampled controllers, outcomes and nominal settings.

The construction connects a variable-width source law to explicit temporal
predictions and shared-resource requirements. Optimal finite-resource
approximation, robust width-grid selection and a derivation from the original
retrocausal source dynamics remain open. In particular, the operational
construction does not determine the hidden source law from which the pair
probabilities originated.

\section{Discussion}\label{sec:discussion}
The results distinguish three requirements that coincide neither mathematically nor operationally: invariance of branch probabilities, invariance of conditional record laws, and realizability of the correction within the specified apparatus. Theorem~\ref{thm:conditional} connects the first two. Repeated passive access reveals higher conditional moments of an unresolved angle; a readable compensating controller exposes conditional information that disappears from its separate marginal. At a junction, the correction must also respect the available factorization. Under reuse, the retained controller changes the conditional instrument even when the isolated calibration is unchanged.

Each model admits a different resolution. Copying a protected bit adds no new observation of the hidden angle. The Bell-junction benchmark defines a normalized contextual law with an exactly characterized correction rank. The Gaussian family admits a sharp unrestricted law-change frontier and an exact classification of the protected parameter choice, while angular-only readout removes its radial information. The shared-controller construction supplies an operational realization of the isolated variable-width family and a positive finite-resource approximation. These resolutions impose different readout, interaction, or resource assumptions; they are not interchangeable implementations of the same hidden-history model.

For approximate composition, the relevant control is on conditional kernels with memory retained, rather than on isolated marginals. The following standard estimate makes this distinction explicit.
\begin{proposition}[Conditional error accumulation]\label{prop:accumulation}
Consider two finite sequential record laws with the same initial preparation. Suppose their normalized conditional next-record kernels differ by at most $\delta_j$ in total variation at stage $j$, uniformly over the relevant earlier histories. Their complete $T$-stage record distributions then satisfy
\begin{equation}
 \TV(P_{1:T},\widetilde P_{1:T})
 \leq \min\!\left\{1,\sum_{j=1}^T\delta_j\right\}.
\end{equation}
\end{proposition}
\begin{proof}
Insert and subtract one kernel at a time in the products of conditional kernels. Integration over unchanged normalized future kernels leaves the variation due to the replaced kernel, bounded by $\delta_j$. Sum these bounds and use $\TV\leq1$.
\end{proof}
The premise cannot be inferred from closeness of single-time distributions. The repeated-readout and retained-controller examples give explicit distinctions between these levels of description. Both reset models in Sec.~\ref{sec:memory} are causal quantum constructions and preserve earlier records; their separation concerns later temporal predictions, not backwards signalling.

Several questions remain. The bandwidth-dependent lower bound need not be sharp, and an optimal tradeoff should identify an attaining response family or a tighter regularity constraint. At the junction, a reusable physical carrier may admit correction rules outside the product-sum representation considered here; approximate resource bounds would be more informative than exact continuum rank alone. For the Gaussian controller, the boundary ensemble and coupling should ultimately be derived from dynamics that also includes every accessible pointer and reset record. Existing phase-space approaches already incorporate measurement dynamics~\cite{drummond}; comparison with such theories requires multitime record laws and intervention rules. Finally, the discrete shared-controller construction gives an attainable approximation rate, but its optimal resource scaling for adaptive instruments is undetermined.

The source-specific results use established ingredients, including antipodal symmetry~\cite{almada}, grouped-sample identifiability~\cite{vandermeulen}, source independence~\cite{branciard}, and operational descriptions of memory~\cite{pollock,taranto}. Their role here is to expose the assumptions needed to extend a protected isolated record. The obstruction statements apply to the specified apparatus classes, while the positive operational construction reproduces an outcome family without identifying its hidden-angle mechanism. This separation is necessary when assessing whether a proposed retrocausal apparatus remains consistent under richer observations and composition.

\section*{Data and code availability}
The simulation code, parameter specifications, numerical summaries, and trial-level records are provided in the accompanying reproducibility archive. Appendix~\ref{app:methods} describes the sampling and verification procedures. All figures are generated from the analytical expressions and simulation records specified in this article.
\paragraph{Author contributions and use of AI tools.}
The author used OpenAI's ChatGPT to assist with the implementation of numerical verification code and the preparation of the manuscript. The author takes responsibility for the final arguments, computations, interpretations, and conclusions.

\clearpage
\onecolumn
\appendix
\section{Numerical methods and reproducibility}\label{app:methods}
\subsection{Sampling and uncertainty}
The numerical study comprises 8,300,000 independent simulated histories. The history, rather than each component outcome, is the independent sampling unit. The repeated-readout experiment uses 18 contexts with $10^5$ histories each, giving 1,800,000 histories and 5,400,000 fresh detector outcomes. Copied records reuse these outcomes and do not increase the independent sample count. The network experiment uses three policies under two laws with $10^5$ histories per group, giving 600,000 histories. The Gaussian-controller experiment uses 11 configurations and two policies: the baseline has 150,000 histories per policy and each other configuration has 100,000, giving 2,300,000 histories. The variable-width study comprises 1,600,000 source trials, 400,000 two-stage trials, and 1,600,000 adaptive trials, giving 3,600,000 histories and 8,800,000 detector outcomes.

Continuous latent variables are sampled from the specified distributions, independently of the quadrature grids used for analytical comparisons. The repeated-readout sampler draws a wrapped-Cauchy candidate and accepts according to the remaining positive source factor. The network sampler draws conditional source angles by rejection after sampling the branch label. The Gaussian sampler draws controller coordinates from the conditional radial distribution, including the stated finite cutoff where applicable. The shared-controller sampler uses the inverse distribution function in Sec.~\ref{sec:memory}. Each experiment retains the sampled latent variables and outcomes used to compute its summary statistics.

Reported 95\% intervals are pointwise unless explicitly described as simultaneous. For a Bernoulli count $k$ from $n$ independent histories, the Wilson interval has centre $c$ and half-width $w$ given by
\begin{equation}
\begin{aligned}
 c&=\frac{\widehat p+z^2/(2n)}{1+z^2/n},\\
 w&=\frac{z\sqrt{\widehat p(1-\widehat p)/n+z^2/(4n^2)}}{1+z^2/n},
\end{aligned}
\end{equation}
where $\widehat p=k/n$ and $z=\Phi^{-1}(0.975)$. In the repeated-readout experiment, exact Clopper--Pearson intervals additionally use Bonferroni adjustment over 90 event probabilities, with tail probability $0.05/(2\times90)$. The primary reset comparison uses the difference of two independent sample proportions and the corresponding normal-approximation standard error. The classifier comparison in Sec.~\ref{sec:controller} is treated analogously. Overlapping event counts within one history are not assumed independent when assigning simultaneous coverage.

Some pointwise intervals exclude their analytical probabilities, as expected across many comparisons. The complete trial records underlying the reported sample counts are retained. The finite parameter sweeps and exact policy enumerations are distinct from the Monte Carlo sample totals and support claims only over their stated domains. Maxima on finite grids give lower bounds on continuum suprema; they are not certificates of global optimization.

\subsection{Deterministic checks and retained data}
Each deterministic calculation is compared with a separately evaluated expression. For repeated readout, integration on 32,768 and 65,536 angular points checks refinement, with maximum reported change $1.82\times10^{-14}$. For the Bell junction, positive-factor quadrature is compared with harmonic branch masses and independent quantum traces; the largest coarse-grid discrepancy is $1.77\times10^{-13}$. For the Gaussian controller, radial integration reproduces the analytical variation with discrepancy at most $1.11\times10^{-16}$ in the stated checks. For the memory model, finite Fourier expressions are compared with quadrature and with density-matrix instruments, with discrepancies of order $10^{-16}$. These values quantify the specified floating-point comparisons, rather than a common rigorous error bound. Exact rank claims follow from their analytical proofs, not from singular values below numerical resolution.

The implementation uses Python 3.12.14, NumPy 2.3.5, SciPy 1.17.0, and Matplotlib 3.10.8. The master seeds for the repeated-readout, network, controller, and memory experiments are, respectively, 2026091801, 2026091802, 2026091803, and 2026091704. Their independent streams and sampling counts are specified in the accompanying code. Each experiment provides a Python script, JSON results, CSV summaries, compressed NPZ trial arrays, and vector figures. Arrays can be read without executable object deserialization using \texttt{numpy.load} with \texttt{allow\_pickle=False}.

The accompanying verification program checks all 51 compressed-array archives, recounts reported events from retained outcomes, checks angle and outcome domains, verifies the copied-record identities, and reconstructs the Gaussian classifier directly from saved coordinates. The aggregate runner executes the four experiments in a new output directory. These numerical checks establish agreement between the retained records and the reported calculations; the theorem statements rely on the proofs given in the main text.

\end{document}